\documentclass[a4paper]{amsart}

\usepackage{stackrel}
\usepackage{graphicx, color} 
\usepackage{hyperref}
\hypersetup{colorlinks=true, linkcolor=blue,citecolor=citeblue}
\definecolor{citeblue}{rgb}{0.2,0.2,0.6}
\usepackage{lmodern} 
\usepackage[normalem]{ulem}
\usepackage{scrextend}
\changefontsizes{11pt}

\definecolor{darkblue}{rgb}{0.1,0.1,0.6}

\newcommand{\soutb}{\bgroup\markoverwith{\textcolor{blue}{\rule[.5ex]{2pt}{1pt}}}\ULon}
\newcommand{\soutr}{\bgroup\markoverwith{\textcolor{red}{\rule[.5ex]{2pt}{1pt}}}\ULon}

\newcommand\rme{\mathrm{e}}

\newcommand{\eg}{{\it e.g.}\,}
\newcommand{\ie}{{\it i.e.}\,}
\newcommand{\cf}{{\it cf.}\,}

\renewcommand\and{\qquad\text{and}\qquad}

\newcommand\sm{\setminus}
\newcommand\Op{\sfH_{\aa,X}}

\newcommand\dl{\delta}

\newcommand{\comm}[1]{}

\newcommand\G{\Gamma}

\renewcommand\aa{\alpha}

\newcommand\s{\sigma}

\newcommand\g{\gamma}

\newcommand\kp{\kappa}

\newcommand\ii{{\mathsf{i}}}

\renewcommand\Im{{\rm Im}\,}

\newcommand\arr{\rightarrow}

\newcommand\dd{{\mathsf{d}}}

\newcounter{counter_a}
\newenvironment{myenum}{\begin{list}{{\rm(\roman{counter_a})}}%
{\usecounter{counter_a}
\setlength{\itemsep}{1.ex}\setlength{\topsep}{0.8ex}
\setlength{\leftmargin}{5ex}\setlength{\labelwidth}{5ex}}}{\end{list}}

\usepackage[latin1]{inputenc}
\usepackage[T1]{fontenc}

\numberwithin{figure}{section}
\numberwithin{equation}{section}
\theoremstyle{plain}
\newtheorem*{thm*}{Theorem}

\newtheorem{thm}{Theorem}[section]

\newtheorem{prop}[thm]{Proposition}

\theoremstyle{remark}
\newtheorem{remark}[thm]{Remark}
\theoremstyle{plain}

\newcommand\ov{\overline}
\newcommand\wt{\widetilde}
\newcommand\wh{\widehat}

\newcommand\sign{{\rm sign\,}}

\def\ov{\overline}

      \def\dC{{\mathbb C}}

   \def\dN{{\mathbb N}}   
      \def\dR{{\mathbb R}}
\def\dS{{\mathbb S}}

   \def\sfH{{\mathsf H}}

\def\sfS{{\mathsf S}}

\def\cD{{\mathcal D}}

   \def\cN{{\mathcal N}}   \def\cO{{\mathcal O}}

   \def\cZ{{\mathcal Z}}

\def\sfm{{\mathsf m}}

\newcommand{\dom}{\mathrm{dom}\,}

\makeatletter
\def\section{\@startsection{section}{1}\z@{.9\linespacing\@plus\linespacing}%
	{.7\linespacing} {\fontsize{13}{14}\selectfont\bfseries\centering}}
\def\paragraph{\@startsection{paragraph}{4}%
	\z@{0.3em}{-.5em}%
	{$\bullet$ \ \normalfont\itshape}}

\@addtoreset{equation}{section}
\makeatother

\newtheorem{claim}{Claim}[section]
\newtheorem{theorem}[claim]{Theorem}
\newtheorem{definition}[claim]{Definition} 

\numberwithin{equation}{section}

\newcommand\vX{\mathsf{v}_X}

\begin{document}
\title[Asymptotics of resonances induced by point interactions]{Asymptotics of resonances induced by point interactions}
\author{Ji\v{r}\'{i} Lipovsk\'{y}}
\address{Department of Physics, Faculty of Science, University of Hradec Kr\'{a}lov\'{e}, Rokitansk\'{e}ho 62, 500\,03 Hradec Kr\'{a}lov\'{e}, Czechia}
\email{jiri.lipovsky@uhk.cz}
\author{Vladimir Lotoreichik}
\address{Department of Theoretical Physics,  Nuclear Physics Institute CAS, 25068 \v{R}e\v{z} near Prague, Czechia}
\email{lotoreichik@ujf.cas.cz}

\begin{abstract}
	We consider the resonances of the self-adjoint 
	three-dimensional Schr\"odinger operator with point interactions of constant strength supported on the set $X = \{ x_n \}_{n=1}^N$. 
	The size of $X$ is defined by 
	$V_X = \max_{\pi\in\Pi_N} \sum_{n=1}^N |x_n - x_{\pi(n)}|$, where $\Pi_N$ is the family of
	all the permutations of the set $\{1,2,\dots,N\}$.
	We prove that the number
	of resonances counted with multiplicities
	and lying inside the disc of radius $R$ behaves
	asymptotically linear $\frac{W_X}{\pi} R  + \cO(1)$ as $R \arr \infty$, where
	the constant $W_X \in [0,V_X]$ can be seen as the effective size of $X$. Moreover, we show that there exist configurations of any number of points such that  $W_X = V_X$. Finally, we construct an example for $N = 4$ with  $W_X < V_X$, which can be viewed as an analogue of a quantum graph
		with non-Weyl asymptotics of resonances.
\end{abstract}

\maketitle


\section{Introduction}
%
In this note we discuss the resonances of
the three-dimensional Schr\"o\-dinger operator $\Op$ with point interactions of constant strength $\aa \in\dR$ supported on the discrete set $X =\{x_n\}_{n=1}^N\subset\dR^3$, $N\ge 2$. 
The corresponding Hamiltonian $\Op$ is associated with the formal differential expression
\begin{equation}\label{eq:formal}
	-\Delta + \aa\sum_{n=1}^N \dl(x-x_n)\,,\qquad \text{on}~\dR^3\,,
\end{equation}
where $\dl(\cdot)$ stands for the
point $\dl$-distribution in $\dR^3$.
The Hamiltonian $\Op$ can be rigorously defined as
a self-adjoint extension of a certain symmetric operator in the Hilbert space $L^2(\dR^3)$; \cf~Section~\ref{sec:operator} for details. Resonances
of $\Op$ were discussed in the monograph~\cite{AGHH}
and in several more recent publications \eg~\cite{AK17, BFT, EGST96}, see also the review \cite{DFT} and the references therein. 

Our ultimate goal is to obtain the asymptotic distribution
for the resonances of $\Op$.
To this aim, we define the \emph{size} of $X$ by
\begin{equation}\label{eq:V}
	V_X 
	:= 
	\max_{\pi\in \Pi_N} 
		\sum_{n=1}^N |x_n - x_{\pi(n)}|\,,
\end{equation}
where $\Pi_N$ is the family of all the permutations
of the set $\{1,2,\dots,N\}$. A~graph-theoretic
interpretation of the value $V_X$ through so-called
irreducible pseudo-orbits is given in Remark~\ref{rem:pseudo_V}. 	 
This definition
of the size is motivated by the condition on resonances for $\Op$ given in Section~\ref{sec:condition}.
As the main result of this note, we prove that the number $\cN_{\aa,X}(R)$ of the resonances of $\Op$ lying inside the disc $\{z\in\dC\colon |z| < R\}$ and with multiplicities taken into account 
behaves asymptotically linear
\begin{equation}\label{eq:asymp}
	\cN_{\aa,X}(R) 
	= \frac{W_X}{\pi} R + \cO(1)\,,
	\qquad R\arr \infty\,,
\end{equation}
where the constant $W_X\in [0,V_X]$
does not depend on $\aa$ and can be viewed as the effective size of $X$. The constant $W_X$ can be computed by an implicit formula~\eqref{eq:formula_WX}. 
It is not at all clear whether a simple explicit
formula for $W_X$ in terms of $X$ can be found.

In the proof of~\eqref{eq:asymp} we
use that the resonance condition for $\Op$ acquires the form of an exponential polynomial, which can be obtained by a direct computation or alternatively
using the pseudo-orbit expansion
as explained in Section~\ref{sec:pseudo}. 
Recall that an exponential polynomial is a sum of finitely many terms, each of which is a product of a rational function and an exponential; 
\cf~the review paper~\cite{La} and the monographs~\cite{BeC,BeGa}.
In order to obtain the asymptotics~\eqref{eq:asymp} we employ a classical result on the distribution of zeros of exponential polynomials, recalled in Section~\ref{sec:exppoly}
for the convenience of the reader. 

A configuration of points $X$
for which $W_X = V_X$ is said to be of \emph{Weyl-type}.
We show that for any $N \in \dN$ there exist
Weyl-type configurations consisting of $N$ points. 
For two and three points ($N\le 3$),
in fact, any configuration is of Weyl-type, 
as shown in Section~\ref{sec:Weyl}.
On the other hand, we present in 
Section~\ref{sec:nonWeyl} an example 
of a non-Weyl configuration for $N = 4$, for which strict inequality $W_X < V_X$  holds. We expect that such configurations can also be constructed for any $N > 4$. One can trace an analogy 
with non-Weyl quantum graphs studied in~\cite{DEL,DP}. Non-uniqueness of the permutation at which the maximum in~\eqref{eq:V}
is attained, is a necessary condition for a configuration of points $X$ to be non-Weyl.
Exact geometric characterization of non-Weyl-type point configurations remains an open question. Besides
that a physical interpretation of this mathematical
observation still needs to be clarified. 

It is worth pointing out that $\cN_{\aa,X}(R)$   is \emph{asymptotically linear} similarly as the counting function for resonances of the one-dimensional 
Schr\"odinger operator $-\frac{\dd^2}{\dd x^2} +V$
with a potential $V\in C^\infty_0(\dR;\dR)$; see~\cite{Z87}. The exact asymptotics of
the counting function for  resonances of the three-dimensional
Schr\"odinger operator $-\Delta + V$ with a
potential $V\in C^\infty_0(\dR^3;\dR)$
is known  only in some special cases,
but for ``generic'' potentials
this counting function behaves as $\sim\!R^3$,
thus being \emph{not asymptotically linear}; see~\cite{CH08} for details.


\section{Exponential polynomials}\label{sec:exppoly}

In this section we introduce exponential polynomials and recall a classical result on the asymptotic distribution of their zeros. This result was first obtained by 
P\'olya~\cite{Po} and later improved by
many authors, including Schwengeler~\cite{Sch} and
Moreno~\cite{M73}. We refer the reader to the 
review~\cite{La} by Langer and to the monographs~\cite{BeC,BeGa}.
\begin{definition}\label{def:exppoly}
	An exponential polynomial $F\colon\dC\arr\dC$ is a function of the form
	\begin{equation}\label{eq-exppol}
		F(z) = 
		\sum_{m=1}^M z^{\nu_m} A_m(z) \rme^{\ii z\s_m}\,, 
	\end{equation}
	where $\nu_m \in \dR$, $m= 1,2,\dots,M$,  $A_m(z)$ are rational functions in~$z$ not vanishing identically, and the constants 
	$\s_m\in \dR$ are ordered increasingly ($\s_{\min}:= \s_1 < \s_2 < \dots < \s_M =: \s_{\max} $). 
\end{definition}
For example, for the exponential polynomial 
\[
	F(z) = 
	z \frac{z+\ii}{z-\ii} \rme^{\ii z} 
	+ 
	z^2 \frac{z^2+\ii}{z^2+ 1} \rme^{2\ii z}
\]
we have $M = 2$, $\nu_1 = 1$, $\nu_2 = 2$,
$\s_1 = 1$, $\s_2 = 2$, $A_1(z) = \frac{z+\ii}{z-\ii}$, $A_2(z) = \frac{z^2+\ii}{z^2+ 1}$.

The zero set of an exponential polynomial $F$ is defined by
\begin{equation}\label{key}
	\cZ_F := \{z\in\dC\colon F(z) = 0\}\,.
\end{equation}
For any $z\in\cZ_F$ we define its multiplicity
$\sfm_F(z)\in\dN$
as the algebraic multiplicity of the root $z$ of the function~\eqref{eq-exppol}.
Moreover, we introduce the counting function for an exponential polynomial $F$ by 
\[
	\cN_F(R) = \sum_{z\in\cZ_F\cap\cD_R}
	\sfm_F(z)\,,
\]
where $\cD_R := \{z\in\dC\colon |z| < R\}$
is the disc in the complex plane centered at the origin
and having the radius $R>0$. Thus, the value $\cN_F(R)$ equals the number of zeros of $F$ counted with multiplicities and lying inside $\cD_R$. 
Now we have all the tools at our disposal to state the result on the asymptotics of $\cN_F(R)$, proven in~\cite[Thm. 6]{La},
see also \cite[Thm. 3.1]{DEL}. 
\begin{theorem}
	\label{lem1}
	Let $F$ be an exponential polynomial 
	as in~\eqref{eq-exppol} such that 
	\[
		\lim_{z\to \infty} A_m(z) = a_m \in\dC\sm\{0\}\,,
		\qquad \forall\, m =1,2,\dots,M\,.
	\] 
	Then the counting function for $F$ asymptotically behaves as
	\[
		\cN_F(R) 
		= 
		\frac{\s_{\max}-\s_{\min}}{\pi}\, R + \cO(1)\,,
		\qquad R\arr \infty\,.
	\]
\end{theorem}
%

\section{Rigorous definition of $\Op$}
\label{sec:operator}
The Schr\"odinger operator $\Op$ associated with the formal differential expression~\eqref{eq:formal} can be rigorously defined as a self-adjoint extension in $L^2(\dR^3)$ of the closed, densely defined, symmetric operator
\begin{equation}\label{key}
	\sfS_X u   := -\Delta u\,,\qquad 
	\dom\sfS_X := \{u\in H^2(\dR^3)\colon u|_X = 0\}\,,
\end{equation}
where the vector $u|_X = (u(x_1),u(x_2),\dots, u(x_N))^{\top}\in\dC^N$
is well-defined by the Sobolev embedding theorem~\cite[Thm. 3.26]{McL}.
The self-adjoint extensions of $\sfS_X$ with $N = 1$
have been first analyzed in the seminal paper~\cite{BF61}.
For $N > 1$ the symmetric operator $\sfS_X$ possesses
a rich family of self-adjoint extensions, not all
of which correspond to point interactions.
The
self-adjoint extensions of $\sfS_X$
corresponding to point interactions are investigated in detail in the monographs~\cite{AGHH, AK}, see also the references therein. Several alternative ways for parameterizing of all the self-adjoint extensions of $\sfS_X$ can be found in a more recent literature; see \eg~\cite{GMZ12, P08, T90}. Below we  follow the strategy of~\cite{GMZ12} and use some of notations therein. According to~\cite[Prop. 4.1]{GMZ12}, the adjoint of $\sfS_X$
can be characterized as follows
\begin{equation*}\label{key}
\begin{split}
	\dom \sfS_X^* & = 
	\left \{
		u = u_0 + \sum_{n=1}^N
		\left (
			\xi_{0n}\frac{\rme^{-r_n}}{r_n} 
						+ \xi_{1n} \rme^{-r_n}
		\right )
		\colon u_0 \in\dom\sfS_X,
		\xi_0,\xi_1\in\dC^N 
	\right \}\,,\\
	\sfS_X^* u  & = -\Delta u_0 - \sum_{n=1}^N
	\left (
	\xi_{0n}
	\frac{\rme^{-r_n}}{r_n} + \xi_{1n}
	\left (\rme^{-r_n} - \frac{2\rme^{-r_n}}{r_n}\right )
	\right )\,,\\
\end{split}
\end{equation*}
where $r_n\colon\dR^3\arr\dR_+$, $r_n(x) := |x-x_n|$ for
all $n =1,2,\dots, N$ and $\xi_0 = \{\xi_{0n}\}_{n=1}^N$, $\xi_1 = \{\xi_{1n}\}_{n=1}^N$.
Next, we introduce the mappings 
$\G_0,\G_1\colon \dom\sfS_X^* \arr \dC^N$ by 
\begin{equation}\label{key}
	\G_0 u := 4\pi\xi_0 \and
	\G_1 u  := \left \{\lim_{x\arr x_n}\left (u(x) - 
	\frac{\xi_{0n}}{r_n}\right ) \right\}_{n=1}^N\,.
\end{equation} 
Eventually, the operator $\Op$ is defined as
the restriction of $\sfS_X^*$
\begin{equation}\label{key}
	\Op u    := \sfS_X^* u\,,\qquad
	\dom\Op  := \left \{u\in\dom\sfS_X^*\colon
	\G_1 u = \aa\G_0 u\right \}\,,
\end{equation}
\cf~\cite[Rem. 4.3]{GMZ12}. 
Finally, by~\cite[Prop. 4.2]{GMZ12}, 
the operator $\Op$ is self-adjoint in $L^2(\dR^3)$.
Note also that the operator $\Op$ is the same as
the one considered in~\cite[Chap. II.1]{AGHH}.
We remark that the usual self-adjoint
free Laplacian in $L^2(\dR^3)$ formally corresponds to the case $\aa = \infty$.

\section{Resonances of $\Op$}\label{sec:resonances}
The main aim of this section is to prove asymptotics of resonances given in~\eqref{eq:asymp}. Apart from that we provide
a condition on resonances through the pseudo-orbit
expansion, which is of independent interest and
which leads to an interpretation of the constant
$V_X$ in the graph theory.

\subsection{A condition on resonances for $\Op$}
\label{sec:condition}
%
First, we recall the definition of resonances for $\Op$ borrowed from~\cite[Sec. II.1.1]{AGHH}. This
definition provides at the same time a way to find them. To this aim we introduce the function
\begin{equation}\label{eq:F}
	F_{\aa,X}(\kp) := \det\left[
	\left \{
		\left(\aa- \frac{\ii \kp}{4\pi}\right)\dl_{nn'}
		-\wt G_\kp(x_n-x_{n'})
	\right\}_{n,n'=1}^{N,N}
	\right]\,, 
\end{equation}
where $\dl_{nn'}$ is the Kronecker symbol
and
 $\wt G_\kp(\cdot)$ is given by
\[
	\wt G_\kp(x) := 
	\begin{cases}
	0\,,&  x = 0\,,\\
	\frac{\rme^{\ii \kp |x|}}{4\pi |x|}\,,
	&  x \ne 0\,.
	\end{cases}
\]	
We say that $\kp_0\in\dC$ 
is a resonance of $\Op$ if
\begin{equation} \label{eq:rescon1}
	F_{\aa,X}(\kp_0) = 0\,,
\end{equation} 
holds. The multiplicity of the resonance $\kp_0$ equals the multiplicity of the zero of 
$F_{\aa,X}(\cdot)$ at $\kp = \kp_0$. 
In our convention
true resonances and negative eigenvalues of 
$\Op$ correspond to $\Im\kp_0 < 0$ and
$\Im\kp_0 > 0$, respectively. According
to~\cite[Thm. II.1.1.4]{AGHH} the number of negative eigenvalues of $\Op$ is finite and in the end it does not contribute to the asymptotics of the counting function for resonances of $\Op$. 
A connection between the above definition
of the resonances for $\Op$ and a more fundamental definition through the poles of the analytic continuation of the resolvent for $\Op$ can be justified
through the Krein formula in~\cite[\S II.1.1, Thm 1.1.1]{AGHH}.

It is not difficult to see using standard formula
for the determinant of a matrix
that $F_{\aa,X}$ is an exponential polynomial as in Definition~\ref{def:exppoly} with the coefficients dependent on $\aa$ and on the set $X$.

\subsection{Asymptotics of the number of resonances}
Recall the definition of the counting function
for resonances of $\Op$.
\begin{definition}
	We define the \emph{counting function} $\cN_{\aa,X}(R)$ as the number of resonances of $\Op$
	with multiplicities lying inside the disc $\cD_R$.
\end{definition}
%
Now, we have all the tools to provide a proof for the asymptotics
of resonances~\eqref{eq:asymp} stated in the introduction.

%
\begin{theorem}\label{thm-asym}
	The counting function for resonances of $\Op$
	asymptotically behaves as
	\begin{equation}\label{eq-asym}
		\cN_{\aa,X}(R) 
			= \frac{W_X}{\pi}R + \cO(1)\,,
		\qquad R\arr+\infty\,,
	\end{equation}
	with a constant $W_X \in [0,V_X]$, where $V_X$ is the size
	of $X$ defined in~\eqref{eq:V}. In addition,
	$W_X$ is independent of $\aa$.
\end{theorem}
\begin{proof}
	The argument relies on the resonance condition~\eqref{eq:rescon1}. 
	Note that the element of the matrix under the determinant in~\eqref{eq:F} located in the $n$-th row and the $n'$-th column is a product
	of a polynomial in $\kp$ and the exponential
	$\exp(\ii \kp \ell_{n n'})$ with $\ell_{nn'} = |x_n-x_{n'}|$. 
	Hence, expanding $F_{\aa,X}$ by means of a standard formula for the determinant,
	we get that each single term in $F_{\aa,X}$ is a product
	of a polynomial in $\kp$ and the exponential
	$\exp(\ii \kp\sum_{n=1}^N \ell_{n\pi(n)})$, where $\pi\in\Pi_N$ is a permutation of the set $\{1,2,\dots, N\}$.

	The term with the lowest multiple of $\ii \kp$ in the exponential is $\left(\aa - \frac{\ii \kp}{4\pi}\right)^N$, \ie~ there is no exponential
	at all and hence $\s_{\min} = 0$. 
	The largest possible multiple of $\ii \kp$ in the exponentials of $F_{\aa,X}$ is $V_X$. Hence, we get $\s_{\max} \le V_X$. The equality
	$\s_{\max} = V_X$ is not always satisfied.
	If the polynomial coefficient by $\exp(\ii\kp V_X)$ vanishes, 	we have strict inequality $\s_{\max} < V_X$. 
	Finally, Theorem~\ref{lem1} yields
	\[
		\cN_{\aa,X}(R) 
		= \cN_{F_{\aa,X}}(R)
		= \frac{W_X}{\pi} R + \cO(1)\,,
		\qquad R\arr \infty\,,
	\]
	with some $W_X \in [0,V_X]$.
	 		
	The term with the	largest multiple of $\ii \kp$ in the exponent can 
	be represented as a product
	$P\left(
		\aa - \frac{\ii \kp}{4\pi}
		\right )
		\exp(\ii\kp\s_{\rm max})$,
	where $P$ is a polynomial with 
	real coefficients of degree $< N$.
	For simple algebraic reasons, 
	if this term does not identically vanish
	as a function of $\kp$
	for some $\aa = \aa_0\in\dR$, then it does not
	identically	vanish in the same sense  for all $\aa\in\dR$. Hence, we obtain by Theorem~\ref{lem1} that $W_X$ is independent of $\aa$.
\end{proof}
The argument in Theorem~\ref{thm-asym} suggests  
the following implicit formula for the constant
$W_X$
\begin{equation}\label{eq:formula_WX}
	W_X = \inf\Big\{w\in [0,\infty) \colon \lim_{t\arr\infty}
	e^{-wt}|F_{\aa,X}(-\ii t)| = 0\Big\}\,,
\end{equation}
where $F_{\aa,X}(\cdot)$ is as in~\eqref{eq:F}.

\begin{remark}\label{rem:nonweyl}
	The proof of Theorem~\ref{thm-asym} gives slightly more, namely
	the case $W_X < V_X$ can occur only if 
	the maximum in the definition~\eqref{eq:V}
	of the size $V_X$ of $X$ is attained
	at more than one permutation, as otherwise
	cancellation of the principal term
	in the exponential polynomial $F_{\aa,X}$
	can not occur.
\end{remark}

\subsection{Pseudo-orbit expansion for the resonance condition}\label{sec:pseudo}

The resonance condition~\eqref{eq:rescon1} can be
alternatively expressed by contributions of the irreducible pseudo-orbits similarly as for quantum graphs~\cite{BHJ,Li4,Li3}. This expression is just  yet
another way how to write the determinant. 
However, in some cases one can easier find the terms of the determinant by studying pseudo-orbits on the corresponding directed graph and, eventually, verify their cancellations.

Consider a complete metric graph $G$ 
having $N$ vertices identified with the respective
points in the set $X$ and connected by $\frac{N(N-1)}{2}$ edges of lengths $\ell_{nn'} = |x_n - x_{n'}|$. To this graph 
we associate its oriented $G'$ counterpart, which is obtained from $G$ by replacing each edge $e$ of $G$ 
($e$ is the edge between the points with indices $n$ and $n'$) by two oriented bonds $b$, $\wh b$ of lengths $|b| = |\wh b| =  \ell_{nn'}$. The orientation of the bonds is opposite; $b$ goes from $x_n$ to $x_{n'}$, whereas $\wh b$ goes from $x_{n'}$ to $x_n$. 
\begin{definition}\label{def:orbits}
	With the graph $G'$ we associate the following concepts.
	\begin{myenum}
		\item [{\rm (a)}] A \emph{periodic orbit} $\g$ in the graph $G'$ is a closed path, which begins and ends at the same vertex, we label it by the oriented bonds, which it subsequently visits $\g = (b_1, b_2, \dots, b_n)$. 
		\item [{\rm (b)}] 
		A \emph{pseudo-orbit} $\wt \g$ is a collection of periodic orbits $\wt \gamma = \left\{\g_1,\g_2,\dots, \g_{n}\right\}$. 
		The number of periodic orbits
		contained in the pseudo-orbit $\wt \gamma$ will be denoted by  
		$|\wt\g|_{\rm o}\in\dN_0$.
		\item [{\rm (c)}] An \emph{irreducible pseudo-orbit} $\bar \g$ is a pseudo-orbit which does not contain any bond more than once. Furthermore, we define
		\[
			B_{\bar \gamma}(\kp)
			= \prod_{b_j \in \bar\gamma} \left(-\frac{\rme^{\ii \kp |b_j|}}{4\pi |b_j|}\right)\,.
		\]
		For $|\bar\gamma|_{\rm o}=0$  we set $B_{\bar\gamma} := 1$.
		We denote by $\bar\cO_{m}$
		the set of all irreducible pseudo-orbits in $G'$
		containing exactly $m\in\dN_0$ bonds.
		Note that the total length of $\ov{\gamma}$ is given 	by $\sum_{b_j\in\ov\gamma}|b_j|$.
	\end{myenum}
\end{definition} 
Note that any permutation $\pi\in\Pi_N$ can be represented as a product of disjoint cycles~\cite[Sec. 3.1]{Bon04}
\[
	\pi = 
	(v_1,v_2,\dots,v_{n_1})\,(v_{n_1+1}, \dots, v_{n_1+n_2})\cdots (v_{n_1 +\dots +n_{m(\pi)-1}+1},\dots,
	v_{n_1+\dots+n_{m(\pi)}})\,,
\]        
where $m(\pi)$ is the number of them, $n_j = n_j(\pi)$ is the length of the $j^{\rm th}$-cycle, and  $n(\pi)$
is the number of cycles in $\pi$ of length one.  In this notation, each parenthesis denotes one cycle and \eg\,for a cycle 	
$(v_1,v_2,\dots,v_{n_1})$ it holds that $\pi(v_1) = v_2$, $\pi(v_2) = v_3$, \dots , $\pi(v_{n_1}) = v_1$.
The permutations $\Pi_N$ are in one-to-one correspondence with irreducible pseudo-orbits
in Definition~\ref{def:orbits} through the decomposition into cycles; 
\cf~\cite[Sec. 3]{BHJ}. Namely, an irreducible pseudo-orbit $\ov\g = \ov\g(\pi)$ consists of 
periodic orbits, each of which is a cycle of $\pi$
in its decomposition, satisfying $n_j(\pi) > 1$.

With these definitions in hands, we can state the following proposition, whose proof is inspired by the proof of~\cite[Thm. 1]{BHJ}.
\begin{prop}\label{prop:pseudo_res}
	The resonance condition $F_{\aa,X}(\kp) = 0$ in~\eqref{eq:rescon1} can be alternatively written as
	\begin{equation}
	\begin{split}
				& \sum_{\pi \in \Pi_N} \sign\pi	\prod_{n=1}^N 
				\left(
					\left(\aa- \frac{\ii \kp}{4\pi}\right)\dl_{n\pi(n)}
					-\wt G_\kp(x_n-x_{\pi(n)}) 
				\right ) \\
		&\qquad\qquad = (-1)^N \sum_{n=0}^N\sum_{\bar\g\in\bar\cO_n} (-1)^{|\bar\g|_{\rm o}} 
			B_{\bar\g}(\kp) \left(\frac{\ii \kp}{4\pi}-\aa\right)^{N-n} = 0\,.				
	\end{split}
	\label{eq:rescon2}
	\end{equation}

\end{prop}
\begin{proof}
	Expanding the determinant in the definition of $F_{\aa,X}$
	we get
	\begin{equation}\label{eq:Fdet}
		F_{\aa,X}(\kp) 
		= \sum_{\pi \in \Pi_N} \sign\pi	\prod_{n=1}^N 
		\left(
			\left(\aa- \frac{\ii \kp}{4\pi}\right)\dl_{n\pi(n)}
			-\wt G_\kp(x_n-x_{\pi(n)}) 
		\right )
		\,.
	\end{equation}
	According to~\cite[Sec 4.1]{BC09}, we have
	$\sign\pi = (-1)^{N+m(\pi)}$.
	Substituting this formula for $\sign\pi$
	into~\eqref{eq:Fdet}, making use of the correspondence
	between irreducible periodic orbits and permutations, the formula $m(\pi) = n(\pi)+ |\bar\g(\pi)|_{\rm o}$, and performing some simple rearrangements, we find
	\[
	\begin{split}
		F_{\aa,X}(\kp) 
		& =
		\sum_{n=0}^N
		\sum_{\stackrel[n(\pi) = N-n]{}{\pi \in \Pi_N}}
		\sign\pi
		\prod_{s=1}^N 
		\left(
		\left(\aa- \frac{\ii\kp}{4\pi}\right)
		\dl_{s\pi(s)}
			-\wt G_\kp(x_s-x_{\pi(s)}) 
		\right ) \\
		& =
		\sum_{n=0}^N
		\sum_{\stackrel[n(\pi) = N-n ]{}{\pi \in \Pi_N}}
		(-1)^{N+n(\pi)}(-1)^{|\bar\g(\pi)|_{\rm o}}
		B_{\ov\g(\pi)}(\kp) 
		\left(\aa - \frac{\ii\kp}{4\pi}\right)^{N-n} \\
		& =
		(-1)^N\sum_{n=0}^N
		\sum_{\ov\g\in\cO_n}
		(-1)^{|\bar\g|_{\rm o}}
		B_{\ov\g}(\kp) 
		\left(\frac{\ii\kp}{4\pi}-\aa\right)^{N-n}\,.
		\qedhere
	\end{split}
	\]
\end{proof}
\begin{remark}\label{rem:pseudo_V}
	In view of Proposition~\ref{prop:pseudo_res},
	the value $V_X$ in~\eqref{eq:V} can be interpreted
	as the maximal possible total length of an irreducible
	pseudo-orbit in the graph $G'$.
\end{remark}

\section{Point configurations of Weyl- and non-Weyl-types}

Recall that a configuration of points 
is said to be
of Weyl-type if $W_X = V_X$ and of non-Weyl-type
if $W_X < V_X$. 
In this section we provide examples for
both types of point configurations and discuss related questions.
For the sake of convenience,
for a configuration of points $X = \{x_n\}_{n=1}^N$
and a permutation $\pi\in\Pi_N$ we define
\[
	\vX(\pi) := \sum_{n=1}^N |x_n - x_{\pi(n)}|\,.
\]
%

\subsection{Weyl-type configurations}\label{sec:Weyl}
First, we show that for low number of points
non-Weyl configurations do not exist.
\begin{prop}
	For $N=2, 3$, $W_X = V_X$ holds for any
	$X = \{x_n\}_{n=1}^N$.
\end{prop}
\begin{proof}
	For $N = 2$, we have $V_X = 2\ell_{12}$.
	From~\eqref{eq:F} and~\eqref{eq:rescon1} we obtain the resonance condition  
	\[
		\left(\frac{\ii \kp}{4\pi}-\aa\right)^2 - \frac{\rme^{2\ii\kp\ell_{12}}}{(4\pi\ell_{12})^2} =  0\,.
	\]
	Obviously, the coefficient at $\rme^{\ii\kp V_X}$   does not identically vanish  and
	the claim follows from Theorem~\ref{lem1}.

	Let $N=3$.
	Without loss of generality we assume that  $\ell_{12}\geq \ell_{23}\geq \ell_{13}$. 
	By triangle inequality we have $\ell_{12} + \ell_{23}+ \ell_{13}\geq 2 \ell_{12}$. 
	The equality is attained only if all three points  belong to a straight line. Hence, we have
	$V_X = \ell_{12} + \ell_{23} + \ell_{13}$,
	which is attained at the cyclic shift,
	having the decomposition
	$\pi = (1,2,3)$.	
	From~\eqref{eq:rescon1} we obtain the resonance condition  
	\[
	\begin{split}
	  &\qquad\qquad \left(\frac{\ii\kp}{4\pi}-\aa\right)^3-
	  \left(\frac{\ii\kp}{4\pi}-\aa\right)
	  f(\kp) +
	 g(\kp)=0\,,\quad \text{where}\\[0.4ex]
		&f(\kp)  :=
		\frac{1}{(4\pi)^2}
		\left( 
		\frac{\rme^{2\ii\kp\ell_{12}}}{(\ell_{12})^2}
	 		  +
	 		  \frac{\rme^{2\ii\kp\ell_{23}}}{(\ell_{23})^2}
	 		  +
	 		  \frac{\rme^{2\ii\kp\ell_{13}}}{(\ell_{13})^2}\right)\,,  \quad
		g(\kp) 
		:= 
	 	\frac{2\rme^{\ii\kp(\ell_{12}+\ell_{23}+\ell_{13})}}{(4\pi)^3\ell_{12}\ell_{23}\ell_{13}}\,.
	\end{split}
	\]
	For simple algebraic reasons, 
	in both cases
	$\ell_{12} + \ell_{23}+ \ell_{13} > 2 \ell_{12}$
	and
	$\ell_{12} + \ell_{23}+ \ell_{13} = 2 \ell_{12}$
	the coefficient at
	$\rme^{\ii\kp V_X}$ does not vanish identically and the claim also follows from Theorem~\ref{lem1}.
\end{proof}

Next, we show that Weyl-type configurations
are not something specific for low number of points
and they can be constructed for any number of them.
\begin{theorem}
	For any $N\ge 2$ there exist
	a configuration of points $X = \{x_n\}_{n=1}^N$
	such that $W_X = V_X$.
\end{theorem}
\begin{proof}
	We provide two different constructions for the cases of even and odd number of points in the set $X$. 
			
	For $N = 2 m$, $m\in \dN$, 
	we choose the configuration $X = \{x_n\}_{n=1}^{2m}$
	as follows. First, we fix
	arbitrary distinct point
	$x_1,x_2,\dots,x_m$ on 
	the unit sphere $\dS^2\subset\dR^3$,
	so that none of them is diametrically opposite
	to the other. Second,
	we select the point $x_{m+k}\in\dS^2$, $k=1,\dots,m$ to be diametrically opposite
	to $x_k$. For simple geometric reasons, we have
	$V_X = 4m$
	and this maximum is attained at the unique permutation $\pi$ having the following decomposition into cycles
	$\pi = (1,m+1)(2,m+2)\dots(m,2m)$.
	In view of Remark~\ref{rem:nonweyl},
	we conclude that $W_X = V_X$.

	For $N = 2 m+1$, $m\in \dN$, we choose the configuration $X = \{x_n\}_{n=1}^{2m+1}$,
	as follows. 
	First,
	we distribute the points $\{x_n\}_{n=1}^{2m}$
	on $\dS^2$	as in the case of even $N$.
	Second, 
	we put the point $x_{2m+1}$  into the center of $\dS^2$. If a permutation
	$\pi\in\Pi_{2m+1}$ does not contain the cycle $(2m+1)$,
	then we have $\vX(\pi) \le 4m$ 
	and the case of equality occurs only for the permutations
	\begin{equation*}\label{key}
	\begin{split}
		&\pi_1  = 
		(1,m+1)(2,m+2)\dots(m-1,2m-1)(m,2m,2m+1),\\
		&\pi_2  = 
		(1,m+1)(2,m+2)\dots(m-1,2m-1)(m,2m+1,2m),\\
		&\pi_3  = (1,m+1)(2,m+2)\dots(m-2,2m-2)(m-1,2m-1,2m+1)
		(m,2m),\\	
		&\pi_4  = (1,m+1)(2,m+2)\dots(m-2,2m-2)(m-1,2m+1,2m-1)
		(m,2m),\\
		&\dots\dots\\
		&\pi_{2m-1}  =
		(2,m+2)\dots(m-1,2m-1)(m,2m)(1,m+1,2m+1), \\
		&\pi_{2m}  =
		(2,m+2)...(m-1,2m-1)(m,2m)(1,2m+1,m+1).
	\end{split}
	\end{equation*}
	If a permutation
	$\pi\in\Pi_{2m+1}$ contains the cycle $(2m+1)$,
	then we again have $\vX(\pi) \le 4m$ and the case of equality
	happens for the unique permutation
	\[
		\pi_{2m+1} = (1,m+1)(2,m+2)...(m,2m)(2m+1)\,.
	\] 
	Hence, we obtain that $V_X = 4m$.
	Moreover, the exponential polynomial $F_{\aa,X}$
	in~\eqref{eq:F} can be written as
	\[
		F_{\aa,X}(\kp) = 
		(-1)^m\frac{4m+4\pi \aa - \ii\kp}{2^{2m}(4\pi)^{2m+1}}
		\rme^{\ii(4m)\kp} + 
		g_0(\kp) + \sum_{l=1}^Lg_l(\kp)  \rme^{\ii\s_l\kp}\,, 
	\]
	where $\s_l \in (0,4m)$ and $g_0, g_l$ are polynomials, $l=1,2,\dots, L$.
	Finally, by Theorem~\ref{lem1} we get
	$W_X = V_X =  4m$.
\end{proof}

\subsection{An example of a non-Weyl-type configuration}\label{sec:nonWeyl}

Eventually, we provide an example of a configuration
of points $X = \{x_n\}_{n=1}^4$ for which $W_X < V_X$
in Theorem~\ref{thm-asym}, since there 
will be a significant cancellation of some terms.

For $a,b,c > 0$, we consider a configuration of points
$X = \{x_n\}_{n=1}^4$, where
\[
	x_1 = (0,0,0)^\top\,,
	\quad x_2 = (a,-b,0)^\top\,,
	\quad x_3 = (a,b,0)^\top\,, 
	\quad
	x_4 = (c,0,0)^\top\,;
\]
see Figure~\ref{fig1}. 
Notice that 
\begin{equation}\label{eq:lengths}
	\ell_{12} = \sqrt{a^2+b^2},\quad
	\ell_{23} = 2b,\quad
	\ell_{34} = \sqrt{(a-c)^2+b^2},\quad
	\ell_{14} = c.
\end{equation}
Let us assume that $b$ and $c$ are sufficiently small in comparison to $a$, being more precise $2b+c<\sqrt{a^2+b^2}+\sqrt{(a-c)^2+b^2}$.
\begin{figure}
\centering
\includegraphics[width=0.85\textwidth]{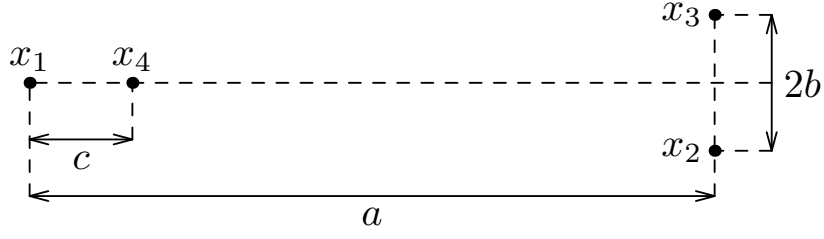}
\caption{Discrete set $X = \{x_n\}_{n=1}^4$ related to example in Section~\ref{sec:nonWeyl}.}
\label{fig1}
\end{figure} 
Let us first write down the general resonance condition~\eqref{eq:rescon1} for four points.
\[
\begin{split}
	& c_0^4 -c_0^2(c_1^2+c_2^2+c_3^2+c_4^2+c_5^2+c_6^2)+2c_0(c_1c_2c_4+c_1c_3c_5+c_2c_3c_6+c_4c_5c_6) \\
	& \qquad\qquad\qquad\quad
	+c_1^2c_6^2+c_2^2c_5^2+c_3^2c_4^2-2(c_1c_2c_5c_6+c_1c_3c_4c_6+c_2c_3c_4c_5) = 0\,,
\end{split} 
\]
where 
\begin{eqnarray*}
	c_0 = \aa-\frac{\ii\kp}{4\pi}\,,
	\quad 
	c_1 = -\frac{\rme^{\ii\kp\ell_{12}}}{4\pi\ell_{12}}\,,
	\quad 
	c_2 = -\frac{\rme^{\ii\kp\ell_{13}}}{4\pi\ell_{13}}\,,
	\quad 
	c_3 = -\frac{\rme^{\ii\kp\ell_{14}}}{4\pi\ell_{14}}\,,\\
	c_4 = -\frac{\rme^{\ii\kp\ell_{23}}}{4\pi\ell_{23}}\,,
	\quad 
	c_5 = -\frac{\rme^{\ii\kp\ell_{24}}}{4\pi\ell_{24}}\,,
	\quad  
	c_6 = -\frac{\rme^{\ii\kp\ell_{34}}}{4\pi\ell_{34}}\,. 
\end{eqnarray*}
In our special case we have 
\begin{equation}\label{eq:lengths2}
	\ell_{12} = \ell_{13},\quad 
	\ell_{34} = \ell_{24} \and 
	\ell_{12}+\ell_{34} > \ell_{14}+\ell_{23}\,.
\end{equation}
Moreover, using~\eqref{eq:lengths} we get
\begin{equation}\label{eq:lengths3}
\begin{split}
	\ell_{12}+\ell_{23}+\ell_{34}+\ell_{14} 
	 & = 2b+c+ \sqrt{a^2+b^2}+\sqrt{(a-c)^2+b^2} \\
	& < 2\sqrt{a^2+b^2}+2\sqrt{(a-c)^2+b^2}  \\
	& = \ell_{12}+\ell_{34}+\ell_{13}+\ell_{24}\,.
\end{split}
\end{equation}
The elements of the group $\Pi_4$ can be decomposed into disjoint cycles as
\begin{center}
\begin{tabular}{ l l l l l  }
$\pi_1 = (1)(2)(3)(4)$, 
&& 
$\pi_9 = (1,2,3)(4)$,
&& 
$\pi_{17} = (1,3)(2,4)$, \\
$\pi_2 = (3,4)(1)(2)$,
&& 
$\pi_{10} = (1,2,3,4)$,
&&  
$\pi_{18} = (1,3,2,4)$,\\
$\pi_3 = (2,3)(1)(4)$,
&&
$\pi_{11} = (1,2,4,3)$,
&&
$\pi_{19} = (1,4,3,2)$,
\\
$\pi_4 = (2,3,4)(1)$,
&&
$\pi_{12} = (1,2,4)(3)$,
&&
$\pi_{20} = (1,4,2)(3)$,
\\
$\pi_5 = (2,4,3)(1)$, 
&&
$\pi_{13} = (1,3,2)(4)$,
&&
$\pi_{21} = (1,4,3)(2)$,
\\
$\pi_6 = (2,4)(1)(3)$,
&&
$\pi_{14} = (1,3,4,2)$,
&&
$\pi_{22} = (1,4)(2)(3)$,\\
$\pi_7 = (1,2)(3)(4)$, 
&&
$\pi_{15} = (1,3)(2)(4)$,
&&
$\pi_{23} = (1,4,2,3)$,
\\
$\pi_8 = (1,2)(3,4)$,
&&
$\pi_{16} = (1,3,4)(2)$,
&&
$\pi_{24} = (1,4)(2,3)$.
\end{tabular}
\end{center}
Using the above decompositions of permutations and~\eqref{eq:lengths2},~\eqref{eq:lengths3} we find
\[
\begin{split}
	\vX(\pi_8) & = 	\vX(\pi_{11}) = \vX(\pi_{14}) = \vX(\pi_{17}) \\
	& > \vX(\pi_{10}) = \vX(\pi_{18}) = \vX(\pi_{19}) = 
	\vX(\pi_{23}) > \dots > \vX(\pi_1) = 0\,.
\end{split}
\]
Hence, $V_X = \vX(\pi_8) = \vX(\pi_{11}) = \vX(\pi_{14}) = \vX(\pi_{17})$ and in view of~\eqref{eq:lengths2}  the leading term corresponding to $\exp(\ii\kp V_X)$ in the resonance condition~\eqref{eq:rescon1} cancels
\[
  \frac{\rme^{2\ii\kp(\ell_{12}+\ell_{34})}}{(4\pi)^4\ell_{12}^2\ell_{34}^2} 
  + \frac{\rme^{2\ii\kp(\ell_{13}+\ell_{24})}}{(4\pi)^4\ell_{13}^2\ell_{24}^2} -\frac{2\rme^{\ii\kp(\ell_{12}+\ell_{34}+\ell_{13}+\ell_{24})}}{(4\pi)^4\ell_{12}\ell_{34}\ell_{13}\ell_{24}}  
  = 0\,.
\]
However, the succeeding term in the condition~\eqref{eq:rescon1}
corresponding to the exponent $\exp(\ii\kp \vX(\pi_{10}))$
does not cancel
\[
	-\frac{2}{(4\pi)^4}
	\left(
		\frac{\rme^{\ii\kp(\ell_{12}+\ell_{23}+\ell_{34}+\ell_{14})}}{\ell_{12}\ell_{23}\ell_{34}\ell_{14}}
		+
		\frac{\rme^{\ii\kp(\ell_{13}+\ell_{23}+\ell_{24}+\ell_{14})}}{\ell_{13}\ell_{23}\ell_{24}\ell_{14}}
	\right)\ne 0\,.
\]
Finally, we end up with
\[
	W_X = \vX(\pi_{10}) = \vX(\pi_{18}) = \vX(\pi_{19}) = 
	\vX(\pi_{23})  < V_X.
\]	

\section{Conclusions}\label{sec:discussion}
In this note, we considered the three-dimensional Schr\"odinger operator
$\Op$
with finitely many point interactions of equal strength 
$\aa\in\dR$ supported on the discrete set $X$.
As the main result, we obtained that the
resonance counting
function for $\Op$
behaves asymptotically linear.
The constant coefficient standing by the linear term in this asymptotics can be seen as the effective size of $X$.

The obtained law of distribution for resonances is very much different from the behaviour of the resonance counting function for the three-dimensional Schr\"odinger operator with a regular potential~\cite{CH08}. 
On the other hand, it resembles 
the corresponding law for one-dimensional
Schr\"odinger operators with regular potentials~\cite{Z87} and for quantum graphs~\cite{DP, DEL}.

We associated a complete directed weighted graph $G'$ with the configuration of points $X$ in a natural way. 
The effective size of $X$ can be estimated from above by
the actual size of $X$ defined as the maximal possible total length of an irreducible pseudo-orbit in the graph $G'$.
An implicit formula for finding the effective size of $X$
was given. We also provided examples showing
sharpness of  our upper bound on the effective size of $X$. Point configurations for which the
effective size of $X$ is strictly smaller than its actual size can be seen as analogues
of `non-Weyl' quantum graphs~\cite{DEL, DP}. The physical experiment which could illustrate mathematical results found in this paper
still awaits realization; one of the possibilities could be to use microwave cavities (see e.g. \cite{DKS}).

\subsection*{Acknowledgments}
J.\,L. acknowledges the support by the grant 15-14180Y of the Czech Science Foundation. V.\,L. acknowledges the support by the grant No. 17-01706S of the Czech Science Foundation (GA\v{C}R).
The authors are grateful to Prof. P. Exner for discussions. The authors thank the referee for useful remarks which helped to improve the paper.
\newcommand{\etalchar}[1]{$^{#1}$}

\end{document}